\documentclass[conference]{IEEEtran}
 \IEEEoverridecommandlockouts

\usepackage[cmex10]{amsmath}
\usepackage{amssymb}
\usepackage{epsfig,epsf,pstricks,pgf}
\usepackage{cite}

\usepackage{enumerate}

\usepackage{amsthm}
\usepackage{psfrag}

\usepackage{times}
\usepackage{epsfig}
\usepackage{amsmath}
\usepackage{amsfonts}
\usepackage{algorithm}
\usepackage{algorithmic}
\usepackage{graphicx}
\usepackage{amssymb}
\usepackage{amstext}
\usepackage{latexsym}
\usepackage{color}
\usepackage{ifthen}
\usepackage{multirow}
\usepackage{verbatim}
\usepackage{array,tabularx}
\usepackage{epsf,pstricks,pgf,pst-node}
\usepackage{color}
\usepackage{url}

\newcommand{\mbf}{\mathbf}
\newcommand{\mcl}{\mathcal}

\newcommand{\und}{\underline}

\newtheorem{theorem}{Theorem}
\newtheorem{lemma}{Lemma}

\theoremstyle{definition}
\newtheorem{definition}{Definition}
\newtheorem{example}{Example}

\IEEEoverridecommandlockouts

\begin{document}


\title{Minimum Cost Multicast with Decentralized Sources}
\author{
\IEEEauthorblockN{Nebojsa Milosavljevic, Sameer Pawar, Salim El Rouayheb, Michael Gastpar and Kannan Ramchandran}
\thanks{N. Milosavljevic, S. Pawar, M. Gastpar and K. Ramchandran are with the Department of Electrical Engineering and Computer Science, University
        of California, Berkeley, Berkeley, CA 94720 USA (e-mail:\{nebojsa,spawar, gastpar, kannanr\}@eecs.berkeley.edu).}
\thanks{S. El Rouayheb is with the Department of Electrical Engineering, Princeton University, Princeton, NJ 08544 USA (e-mail: salim@princeton.edu).}
\thanks{M. Gastpar is also with the School of Computer and Communication Sciences, EPFL, Lausanne, Switzerland (e-mail: michael.gastpar@epfl.ch).}

\thanks{This research was funded by the NSF grants
(CCF-0964018, CCF-0830788), a DTRA grant (HDTRA1-09-1-0032), and in part by an
AFOSR grant (FA9550-09-1-0120).}
}


\maketitle


\begin{abstract}
In this paper we study the multisource multicast problem where every sink in a given directed acyclic graph is a client and is interested in a common file.
We consider the case where each node can have partial knowledge about the file as a side information.
Assuming that nodes can communicate over the capacity constrained links of the graph,
the goal is for each client to gain access to the file, while minimizing some linear cost function of
number of bits transmitted in the network. We consider three types of side-information settings: 
(ii) side information in the  form of
linearly correlated packets; and (iii) 
the general setting where the side information at the nodes have an arbitrary (i.i.d.) correlation structure.
In this work we 1) provide a polynomial time feasibility test, i.e., whether or not all the clients can recover the file, and
2) we provide a polynomial-time algorithm that finds the optimal rate allocation among the links of the graph,
and then determines an explicit transmission scheme for cases (i) and (ii).
\end{abstract} 
\section{Introduction}\label{sec:intro}
We consider a multi-source multicast problem, over a given network topology with capacity constrained links. There are two types of nodes in the network; {\em clients} that are interested in recovering the whole content, and source nodes that may posses possibly correlated side-information. To further illustrate the problem set-up consider the following example.



\begin{figure}[h]
\begin{center}
\includegraphics[scale=0.5]{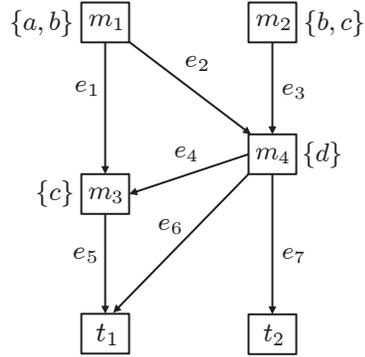}
\end{center}
\vspace{-0.2in}
\caption{An example of the multisource multicast problem, where nodes $m_1,\ldots,m_4$ observe the subsets of the file packets
         $\{a,b,c,d\}$ as shown above. Assuming that nodes can communicate reliably over the capacity constrained links,
         the goal is for the clients $t_1$ and $t_2$ (sinks of the graph) to gain access to the entire file
         while minimizing the communication cost.}\label{fig:twousers}
\end{figure}

A file consists of four equally sized packets
$a$, $b$, $c$ and $d$ belonging to some finite field $\mathbb{F}_{q^n}$.
Also, suppose that the data packets are distributed across the nodes, $m_1$ through $m_4$, that are connected as shown in Figure~\ref{fig:twousers}. The clients denoted by $t_1$ and $t_2$ are interested in recovering the entire file. The edges in the graph are denoted by $e_1,\ldots,e_7$ as shown in Figure~\ref{fig:twousers}. The objective is to minimize some function of the communication cost such that the clients $t_1$ and $t_2$ can recover the entire file. For instance, it can be shown that the following coding scheme minimizes the total number of bits communicated: node~$m_1$ transmits $a$ on link~$e_2$, node~$m_2$
transmits $b,c$ on link~$e_3$, node~$m_3$ transmits $c$ on link~$e_5$, node~$m_4$ transmits $a,b,d$ on link~$e_6$ and $a,b,c,d$
on link~$e_7$.


Note that the example above considers a simple form of the side-information, where different nodes observe partial uncoded
or ``raw'' data packets of the original file. Another important special case of side-information is when
nodes observe linear combinations of the data packets of the original file. In a more general setting the side-information can be of more complex form \emph{i.e.}, arbitrary correlations.

The multisource multicast problem was originally studied by Ho, \emph{et al.}~\cite{HKMESK06}, where for the linearly coded packets the authors
showed under what conditions it is possible to recover the file at all destinations.
For the case of uncoded packets, it is easy to show that one can add a super source as in~\cite{tajbakhsh2011generalized} to the network and then using results from~\cite{lun2006minimum}, find an optimal solution that minimizes the communication cost. In \cite{CXW10,gonen2012coded} the authors considered
a related problem over an undirected graph where all the nodes are interested in recovering the complete file. In \cite{CXW10} it was shown that the problem is NP-hard, while an approximate solution is provided in \cite{gonen2012coded}. In~\cite{haeupler2012network}, Haeupler \emph{et al.} analyzed gossip based protocols in networks where each node observes correlated data.

In this paper, we  make the following contributions.
\begin{itemize}
\item In the case of most general scenario of arbitrarily correlated side information,
we provide conditions as well as a polynomial time algorithm to determine when a multisource multicast problem over a given\texttt{} directed acyclic graph (DAG) is feasible.
\item Using \emph{submodular flow} techniques, we provide a \emph{deterministic polynomial time}
algorithm to find number of bits each node should transmit in order to recover the file at all the clients and be optimal
w.r.t. any linear cost function\footnote{Linear
cost function is defined w.r.t. the number of bits transmitted on each link.}.
\item For the special case of linearly correlated side information we provide an optimal communication scheme based
on the algebraic network coding framework.
\end{itemize}

\section{System Model and Preliminaries}\label{sec:model}

In this work we represent the network by a directed acyclic graph $G=(\mcl{V},\mcl{E})$, where $\mcl{V}$ is the set of nodes,
and $\mcl{E}$ is the set of links that have capacity constraints. We define the capacity function $c:\mcl{E} \rightarrow \mathbb{R}$ 
to denote the maximum number of bits that can be transmitted over a given link. We distinguish between two types of nodes:
1) source nodes $\mcl{M}=\{m_1,m_2,\ldots,m_l\}$ that have partial information about the file, and  
2) clients $\mcl{T}=\{t_1,t_2,\ldots,t_k\}$ which are interested in recovering the file, and are sinks in the graph $G$.
Let $X_{m_1},X_{m_2},\ldots,X_{m_l}$, denote the components of a discrete
memoryless multiple source (DMMS) with a given joint probability mass function. Each source node $m_i \in \mcl{M}$ observes $n$ i.i.d.
realizations of the corresponding random variable $X_{m_i}$, denoted by $X_{m_i}^n$.
We note that the results of this paper can be applied in a straightforward manner when the clients have side information as well.
For the sake brevity, we focus on the case when clients have no side information.

The goal is for each client in $\mcl{T}$ to gain access to all source nodes' observations, \emph{i.e.}, to download the file.
In order to achieve this goal, each node $m_i \in \mcl{M}$ is allowed to send information
across the graph $G$ at rate which is limited by the capacity of the outgoing links of that node.
Transmission of each source node is a function of its own initial observation and all information it receives
from its neighbors. Let us denote transmission on the link $e=(m_i,m_j) \in \mcl{E}$ by
\begin{align}
F_e = f_e\left(X_{m_i}^n, \left\{ F_a : \forall m_r,~\text{s.t.}~a=(m_r,m_i) \in \mcl{E} \right\}\right),
\end{align}
where $f_{e}(\cdot)$ is a mapping of the observations $X^n_{m_i}$ and transmissions received from
the neighbors of $m_i$, $\{m_r : (m_r,m_i) \in \mcl{E} \}$ to an outgoing message on the link $e$.

We denote by $\mcl{M}_{t_i} \subseteq \mcl{M}$ the set of source nodes
which are connected to the client $t_i \in \mcl{T}$. In other words, there exists a path in graph $G$ from every node
in $\mcl{M}_{t_i}$ to the client $t_i \in \mcl{T}$. Consequently, we define the graph $G_{t_i}=(\mcl{V}_{t_i},\mcl{E}_{t_i})$ to be a
subgraph of $G$, where $\mcl{V}_{t_i}=\{\mcl{M}_{t_i},t_i\}$, and $\mcl{E}_{t_i} \subseteq \mcl{E}$ is a set of links that
connects all nodes in $\mcl{M}_{t_i}$ among themselves and with client $t_i$. Furthermore, we assume that
\begin{align}
H\left(X_{\mcl{M}_{t_1}}\right) = \cdots = H\left(X_{\mcl{M}_{t_k}}\right) = H\left(X_{\mcl{M}}\right), \label{cond1}
\end{align}
where $X_{\mcl{M}_{t_i}}\triangleq\left(X_{m_j} : m_j \in \mcl{M}_{t_i} \right)$, and $X_{\mcl{M}}\triangleq \left(X_{m_1},\ldots,X_{m_l}\right)$.
Equality~\eqref{cond1} ensures that every client in the network can potentially gain access to the entire process $X_{\mcl{M}}$.

For each client $t_i \in \mcl{T}$ to learn the file, transmissions $Fe$, $\forall e \in \mcl{E}$, must satisfy,
\begin{align}
\lim_{n \rightarrow \infty} \frac{1}{n} H\left(X_{\mcl{M}}^n| \{F_{e}\}_{e=(m_j,t_i)\in \mcl{E}}\right) = 0,~~~\forall t_i \in \mcl{T}. \label{eq:decode}
\end{align}

\begin{definition}
A rate tuple $\mbf{R}=(R_e : e \in \mcl{E})$ is an {\em achievable multisource multicast (MM) rate vector} if there exists a communication scheme with transmitted messages $\mbf{F}=(F_e : e\in \mcl{E})$ that satisfies \eqref{eq:decode}, and is such that
\begin{align}
R_e = \lim_{n \rightarrow \infty} \frac{1}{n} H(F_e),~~~\forall e \in \mcl{E},
\end{align}
where  $R_e\leq c_e$, $\forall e \in \mcl{E}$.
\end{definition}

In this work, we design a polynomial time algorithm for the multisource multicast problem that minimizes
the linear cost function $\sum_{e \in \mcl{E}} \alpha_e R_e$, where
$\und{\alpha} \triangleq (\alpha_e : e\in \mcl{E}),~0 < \alpha_e < \infty$, $\forall e \in \mcl{E}$, is a vector of non-negative finite weights.
We allow $\alpha_e$'s to be arbitrary non-negative constants, to account for the case when communication
across some group of links in $G$ is more expensive compared to the others. 
Thus, the problem can be formulated as:
\begin{align}
\min_{\mbf{R}} \sum_{e \in \mcl{E}} \alpha_e R_e,~\text{s.t. $\mbf{R}$ is an achievable \emph{MM}-rate vector.} \label{problem1}
\end{align}

\subsection{Finite Linear Source Model}

Now, we briefly describe a special case of a DMMS called the finite linear source model~\cite{CZ10}.
Let $q$ be some power of a prime. Consider the $N$-dimensional random vector $\mathbf{W} \in \mathbb{F}^N_{q^n}$
whose components are independent and uniformly distributed over the elements of $\mathbb{F}_{q^n}.$
Then, in the linear source model, the observations of the nodes $m_i \in \mcl{M}$ is simply given by
\begin{align}
\mathbf{X}_{m_i} = \mathbf{A}_{m_i} \mbf{W},~~\ m_i \in \mcl{M}, \label{model:eq1}
\end{align}
where $\mathbf{A}_{m_i} \in \mathbb{F}_{q}^{\ell_i \times N}$ is the observation matrix of node~$m_i$.

It is easy to verify that for the finite linear source model,
\begin{align}
\frac{H(X_{m_i})}{\log q^n} = \text{rank}(\mbf{A}_{m_i}). \label{rank_entropy}
\end{align}
For the finite linear source model, besides the optimal \emph{MM}-rate vector, we provide a polynomial time code construction based on the algebraic network coding approach~\cite{KM03}.  
\section{Multisource Multicast Rate-Flow Region} \label{sec:flowrate}
In order to solve the optimization problem in~\eqref{problem1} we first establish a region called a ``rate-flow region'' that contains all possible optimal rate allocations. To identify this rate-flow region for our example of Figure~\ref{fig:twousers}, in the case of arbitrarily correlated side-information at the source nodes, we start by considering a single client $t_1$. Next, we isolate the subgraph $G_{t_1}=(\mcl{V}_{t_1},\mcl{E}_{t_1})$ corresponding to $t_1$ and modify its link capacities to infinity as shown in Figure~\ref{fig:snguser}.

%
%
%

\begin{figure}[h]
\begin{center}
\includegraphics[scale=0.5]{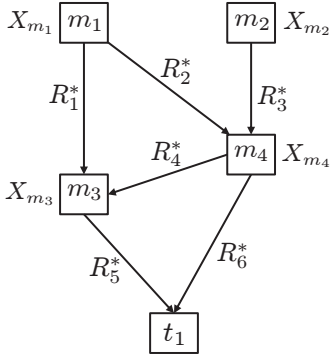}
\end{center}
\vspace{-0.2in}
\caption{Single client multisource multicast problem over graph $G_{t_1}=(\mcl{V}_{t_1},\mcl{E}_{t_1})$ derived from
         the graph $G(\mcl{V},\mcl{E})$ of Figure~\ref{fig:twousers}, for the case of arbitrarily correlated side-information at the source nodes.}\label{fig:snguser}
\end{figure}

Suppose the optimal solution w.r.t. problem~\eqref{problem1} is achieved by $\mbf{R}^{*}=(R^{*}_1,\ldots,R^{*}_6)$.
Then, it follows that transmissions of node~$m_2$ have to satisfy
\begin{align}
R^{*}_3&\geq H(X_{m_2}|X_{m_1},X_{m_3},X_{m_4}), \\
R^{*}_1+R^{*}_2+R^{*}_3 &\geq H(X_{m_1},X_{m_2}|X_{m_3},X_{m_4}). \nonumber
\end{align}
Let us now consider node~$m_4$. Its transmission includes information received from nodes $m_1$ and $m_2$
combined with its own side information. Since the goal is to minimize
the total communication cost, it follows that for the optimal \emph{MM}-rate vector $\mbf{R}^{*}$,
transmission of nodes $m_1$ and $m_2$ cannot be further compressed  at node~$m_4$. Therefore, the transmission of node~$m_4$
consists of 2 components: 1) routed information from nodes~$m_1$ and $m_2$, and 2)Innovative side-information at node
$m_4$ w.r.t. all other source nodes in the network.
Hence, $\mbf{R}^{*}$ must satisfy
\begin{align}
R^{*}_4+R^{*}_6-R^{*}_2-R^{*}_3\geq H(X_{m_4}|X_{m_1},X_{m_2},X_{m_3}).
\end{align}
In order for client $t_1$ to recover the file, \emph{i.e.}, to gain access to $X_{\mcl{M}_{t_1}}$, the incoming links to $t_1$ necessarily
have to carry entire information about the process. In other words
\begin{align}
R^{*}_5+R^{*}_6 = H(X_{\mcl{M}_{t_1}}),
\end{align}
where the equality sign comes from the fact that the goal is to minimize the overall communication cost, and thus,
it is wasteful for client $t_1$ to receive at rate larger than the joint entropy of the process.

Considering all possible subsets of the source node set $\mcl{M}_{t_1}$, we have that an optimal \emph{MM}-rate vector $\mbf{R}^{*}$
must belong to the following rate-flow region
\begin{align}
\partial \mcl{R}_{t_1} &= \{ \partial \mbf{R} : \partial R(\mcl{S})\geq H(X_{\mcl{S}}|X_{\mcl{M}_{t_1} \setminus \mcl{S}}),~\forall \mcl{S}\subset \mcl{M}_{t_1}, \nonumber \\
                 &~~~~~~~~~~~~\partial R(\mcl{M}_{t_1}) = H(X_{\mcl{M}_{t_1}})\}, \label{flow_region}
\end{align}
where
\begin{align}
\partial R(\mcl{S}) \triangleq \sum_{e \in \Delta^{+} \mcl{S}} R_e - \sum_{e \in \Delta^{-} \mcl{S}} R_e,
\end{align}
and
$\Delta^{+}\mcl{S}\subseteq \mcl{E}_{t_1},~(\Delta^{-}\mcl{S}\subseteq \mcl{E}_{t_1})$ denotes the set of links leaving (entering) $\mcl{S}$. For instance, if $\mcl{S}=\{m_3,m_4\}$, then the optimal rate vector $\mbf{R}^{*}$ satisfies
\begin{align}
\partial R^{*}(S)&=R^{*}_5+R^{*}_6-R^{*}_1-R^{*}_2-R^{*}_3 \nonumber \\
             &\geq H(X_{m_3},X_{m_4}|X_{m_1},X_{m_2}).
\end{align}

It can be verified that any rate vector that belongs to the rate-flow region $\partial \mcl{R}_{t_1}$
can be achieved using multi-terminal Slepian-Wolf random-binning scheme~\cite{CT06}. Thus, the rate-flow region $\partial \mcl{R}_{t_1}$
contains all optimal \emph{MM}-rate vectors w.r.t. the optimization problem~\eqref{problem1}.

Extension of this result to a multiple client case is straightforward: an optimal \emph{MM}-rate vector has to
simultaneously belong to all rate-flow regions $\partial \mcl{R}_{t_i}$ which correspond to the graph $G_{t_i}$, $\forall t_i \in \mcl{T}$.
Hence, the optimization problem~\eqref{problem1} can be written as
\begin{align}
&\min_{\mbf{R}} \sum_{e \in \mcl{E}} \alpha_e R_e, \label{problem2}  \\
&~~~~~~\text{s.t.}~\partial \mbf{R} \in \partial \mcl{R}_{t_1} \cap \partial \mcl{R}_{t_2} \cap \cdots \cap \partial \mcl{R}_{t_k}, \nonumber \\
&~~~~~~~~~~~~R_e\leq c_e,~\forall e \in \mcl{E}. \nonumber
\end{align}

Before we address the question of efficiently solving the problem~\eqref{problem2}, first we need to answer whether
or not the problem is feasible.

\section{Feasibility of the Multisource Multicast Problem} \label{sec:feasible}

As in Section~\ref{sec:flowrate}, first, we consider a single client case, \emph{i.e.}, when $\mcl{T}=\{t_1\}$. Then, the obtained result
naturally extends to the setting with arbitrary number of clients.
Before we go any further,
let us introduce some concepts from the combinatorial optimization theory which will turn out to be useful in proving our results.
The set function $f:2^{\mcl{M}_{t_1}}$ is supermodular if
\begin{align}
f(\mcl{S})+f(\mcl{T})\leq f(\mcl{S} \cup \mcl{T}) + f(\mcl{S} \cap \mcl{T}),~\forall \mcl{S},\mcl{T}\ \subseteq \mcl{M}_{t_1}. \label{supermod}
\end{align}
If the inequality sign in~\eqref{supermod} is reversed, then the function $f$ is called submodular.
Let us define the \emph{polyhedron} $P(f)$ and the \emph{base polyhedron} $B(f)$ of a supermodular function $f$ as follows.
\begin{align}
P(f) & \triangleq \{\mbf{Z}~|~\mbf{Z} \in \mathbb{R}^m,~\forall \mcl{S}\subseteq \mcl{M}_{t_1} : Z(\mcl{S})\geq f(\mcl{S}) \}, \label{f:poyh} \\
B(f) & \triangleq \{\mbf{Z}~|~\mbf{Z} \in P(f),~Z(\mcl{M})=f(\mcl{M})_{t_1}\} \label{f:base},
\end{align}
where  $Z(\mcl{S})=\sum_{i \in \mcl{S}} Z_i$. Analogously, we define the polyhedron and the base polyhedron of
a submodular function (the inequality signs in~\eqref{f:poyh} and~\eqref{f:base} are reversed).

It is easy to show that function
\begin{align}
g_{t_1}(\mcl{S}) = H(X_{\mcl{S}}|X_{\mcl{M}_{t_1} \setminus \mcl{S}}),~\forall \mcl{S}\subseteq \mcl{M}_{t_1} \label{fcn:gt}
\end{align}
is supermodular. Hence, the rate-flow region $\partial \mcl{R}_{t_1}$ defined in~\eqref{flow_region} represents the base polyhedron of the function $g_{t_1}$.

\begin{lemma}\label{lm:feasible}
For a single client multisource multicast problem over $G_{t_1}=(\mcl{V}_{t_1},\mcl{E}_{t_1})$, where $\mcl{V}_{t_1}=\{\mcl{M}_{t_1},t_1\}$,
there exists an achievable MM-rate vector, \emph{i.e.} $\partial \mcl{R}_{t_1} \neq \emptyset$, and
$R_e\leq c_e$, $\forall e \in \mcl{E}_{t_1}$, if and only if
\begin{align}
c(\Delta^{+} \mcl{S})\geq H(X_{\mcl{S}}|X_{\mcl{M}_{t_1} \setminus \mcl{S}}),~~\forall \mcl{S}\subseteq \mcl{M}_{t_1}, \label{ineq:feasible}
\end{align}
where
\begin{align}
c(\Delta^{+} \mcl{S}) = \sum_{e \in \Delta^{+} \mcl{S}} c_e,~~ \Delta^{+} \mcl{S} \in \mcl{E}_{t_1}. \nonumber
\end{align}
\end{lemma}
\begin{proof}
As we discussed in Section~\ref{sec:flowrate}, the incoming links to $t_1$ carry entire information
about the process. This combined with the fact that the goal is to minimize the communication cost,
implies that for any optimal \emph{MM}-rate vector $\mbf{R}^{*}$ it holds that
\begin{align}
\sum_{e=(m_j,t_1)\in \mcl{E}_{t_1}} R^{*}_e = H(X_{\mcl{M}_{t_1}}). \label{rate_eq1}
\end{align}
Therefore, without loss of generality we can assume that the capacities of the links incoming to $t_1$ satisfy
\begin{align}
\sum_{e=(m_j,t_1)\in \mcl{E}_{t_1}} c_e = H(X_{\mcl{M}_{t_1}}),
\end{align}
provided that the feasible rate-flow region exists.
It is not hard to show that the capacity function $c(\Delta^{+}\mcl{S})$, $\forall \mcl{S} \subseteq \mcl{M}_{t_1}$ is submodular  (see Chapter~2 in~\cite{F05}).
Let us denote by $\partial \Psi$, the set of the boundaries $\partial \mbf{R}$ of a feasible rate-flow region:
\begin{align}
\partial \Psi \triangleq \{\partial \mbf{R} : R_e \leq c_e,~\forall e \in \mcl{E}_{t_1}\}
\end{align}
In~\cite{hoffman1958some} it was shown that
\begin{align}
\partial \Psi = B(c(\Delta^{+})). \label{base:eq}
\end{align}
From~\eqref{base:eq} and~\eqref{problem2} it follows that there exists a feasible CO rate vector iff
\begin{align}
B(c(\Delta^{+})) \cap B(g_{t_1}) \neq \emptyset. \label{base:intersection}
\end{align}
Problem~\eqref{base:intersection} is known as a \emph{common base problem}~\cite{F05} for which the solution exists if and only if
\begin{align}
c(\Delta^{+}\mcl{S})\geq g_{t_1}(\mcl{S}),~~\forall \mcl{S} \subseteq \mcl{M}_{t_1}.
\end{align}
This completes the proof of Lemma~\ref{lm:feasible}.
\end{proof}
To verify whether there exists an achievable \emph{MM}-rate vector it is necessary to check whether all $2^{|\mcl{M}_{t_1}|}$
inequalities in~\eqref{ineq:feasible} are satisfied. Verifying this is, in general, exponentially hard (in number of nodes).
However, due to the supermodularity of the function $g_{t_1}$, the existence of a common base, and thus the
feasibility of the multisource multicast problem, can be verified
in polynomial time\footnote{Complexity of the \emph{common base} algorithm proposed in~\cite{lawler1982computing} is $\mcl{O}(|\mcl{E}_{t_1}|^3)$} (see~\cite{lawler1982computing} and~\cite{F05}, Chapter 4). This algorithm also provides an achievable \emph{MM}-rate vector (given that it exists)
that belongs to the rate-flow region $\partial \mcl{R}_{t_1}$.

Extensions of the result of Lemma~\ref{lm:feasible} to the case with arbitrary number of clients is straightforward.
We just need to check if the inequalities~\eqref{ineq:feasible} are satisfied for all clients in $\mcl{T}$.
\begin{theorem}
For the multisource multicast problem over $G(\mcl{V},\mcl{E})$, with the capacity function $c$, there exists an achievable
MM-rate vector if and only if
\begin{align}
&c(\Delta^{+} \mcl{S})\geq H(X_{\mcl{S}}|X_{\mcl{M}_{t_i} \setminus \mcl{S}}), \\
&\forall \mcl{S}\subseteq \mcl{M}_{t_i},~\partial \Delta^{+}\mcl{S}\in \mcl{E}_{t_i},~\forall t_i\in \mcl{T}. \nonumber
\end{align}
\end{theorem}
From~\cite{lawler1982computing}, \emph{the common base problem}, and hence the feasibility of the multisource multicast problem can be verified in $\mcl{O}(k\cdot |\mcl{E}|^3)$ time.

\section{Finding the Optimal \emph{MM}-Rates w.r.t.the \\ Linear Communication Cost} \label{sec:alg}

In this section we propose a polynomial time deterministic algorithm that solves optimization problem~\eqref{problem2}.
As in Section~\ref{sec:feasible}, we begin by considering a single client case, \emph{i.e.}, when $\mcl{T}=\{t_1\}$.
Then, by using a similar methodology as in~\cite{lun2006minimum}, we extend our solution to the arbitrary number of clients.

\subsection{Deterministic Algorithm for the Single Client Case} \label{sec:singleuser}

When $\mcl{T}=\{t_1\}$, then, the optimization problem~\eqref{problem2} can be written as
\begin{align}
&\min_{\mbf{R}} \sum_{e \in \mcl{E}_{t_1}} R_e, \label{problem3}  \\
&~~~\text{s.t.}~\partial \mbf{R} \in B(g_{t_1}),~~R_e\leq c_e,~~\forall e \in \mcl{E}_{t_1}, \nonumber
\end{align}
where the supermodular set function $g_{t_1}$ is defined in~\eqref{fcn:gt}.
Next, we introduce the dual set functions. For the function $g_{t_1}$, its dual function $f_{t_1}$ can be obtained as follows:
\begin{align}
f_{t_1}(\mcl{S})=g_{t_1}(\mcl{M}_{t_1})-g_{t_1}(\mcl{M}_{t_1} \setminus \mcl{S}),~~\forall \mcl{S}\subseteq \mcl{M}_{t_1}. \label{fcn:dual}
\end{align}
Applying formula~\eqref{fcn:dual}, we obtain $f_{t_1}=H(X_{\mcl{S}})$ which is a submodular function.
Moreover, in~\cite{F05} it was shown that $B(g_{t_1})=B(f_{t_1})$. Hence we can replace $B(g_{t_1})$
with $B(f_{t_1})$ in~\eqref{problem3}.

Optimization problem~\eqref{problem3} has a form of the \emph{minimum cost submodular flow problem} (see~\cite{F05} for formal definitions),
but with a few differences listed bellow.
\begin{enumerate}
\item In the submodular flow problem, function $g_{t_1}$ has to be defined over all vertices $\mcl{V}_{t_1}$ of graph $G_{t_1}$. However, in our case
      $g_{t_1}$ is a set function over the source vertices only.
\item In the submodular flow problem, $g_{t_1}(\mcl{V}_{t_1})$ must evaluate to $0$, whereas in our problem function
      $g_{t_1}$ is not defined for $\mcl{V}_{t_1}$.
\end{enumerate}

The first step of solving the problem~\eqref{problem3} efficiently involves verifying its feasibility. From
the common base algorithm we obtain an achievable \emph{MM}-rate vector that belongs to $B(f_{t_1})$ provided that $B(f_{t_1}) \neq \emptyset$.
Given any achievable \emph{MM}-rate vector that belongs to $B(f_{t_1})$, one can construct the auxiliary network over graph $G_{t_1}$\footnote{
See Chapter~III of~\cite{F05} for detailed explanation.}.
It can be verified that from this step onwards, we can apply min-cost submodular flow algorithm~\cite{F05} which involves finding negative
cycles of the auxiliary network, and updating the network accordingly along with the achievable \emph{MM}-rate vector.
Comparison between different minimum cost submodular flow algorithms is provided in~\cite{fujishige2000algorithms}.

\subsection{Deterministic Algorithm for the Multiple Client Case}

In this section we extend the results from the previous section to the case where the set $\mcl{T}$
contains arbitrary number of clients. Motivated by the results from~\cite{lun2006minimum}, the
optimization problem~\eqref{problem2} can be written as follows

\begin{align}
&\min_{\mbf{Z},\mbf{R}} \sum_{e \in \mcl{E}} \alpha_e Z_e, \label{problem4}  \\
&~~~~~~\text{s.t.}~Z_e \geq R_e^{(t_i)},~~\forall t_i \in \mcl{T},~\forall e \in \mcl{E}_{t_i}, \nonumber \\
&~~~~~~~~~~\partial \mbf{R}^{(t_i)} \in \partial \mcl{R}_{t_i},~R^{(t_i)}_e\leq c_e,~\forall e \in \mcl{E}_{t_i},~\forall t_i \in \mcl{T}, \nonumber
\end{align}
where $\partial \mcl{R}_{t_i}$ is defined in~\eqref{flow_region} for $i=1$. Equivalence between the optimization problems~\eqref{problem2}
and \eqref{problem4} follows from the fact that transmissions on graph $G$ have to be such that all clients in $\mcl{T}$
learn the file simultaneously.

Optimization problem~\eqref{problem4} has an exponential number of constraints, which makes it challenging to solve in polynomial time.
To obtain a polynomial time solution we consider the Lagrangian dual of problem~\eqref{problem4}.

\begin{align}
&\max_{\mbf{\Lambda}} \sum_{l=1}^k \varphi^{(t_i)} (\mbf{\Lambda}^{(t_i)}), \label{dual} \\
&~~~\text{s.t.}~\sum_{i=1}^k \lambda_e^{(t_i)}=\alpha_e,~\lambda_e^{(t_i)}\geq 0,~\forall t_i \in \mcl{T},~~\forall e \in \mcl{E}_{t_i}, \nonumber
\end{align}
where
\begin{align}
&\varphi^{(t_i)}(\mbf{\Lambda}^{(t_i)}) = \min_{\mbf{R}^{(t_i)}} \sum_{e \in \mcl{E}_{t_i}} \lambda_e^{(t_i)} R_e^{(t_i)}, \label{problem5} \\
&~~~~~~~~~~~~~~~~~~~\text{s.t.}~\partial \mbf{R}^{(t_i)}\in \partial \mcl{R}_{t_i},~~R^{(t_i)}_e\leq c_e,~~\forall e \in \mcl{E}_{t_i}. \nonumber
\end{align}

For any given $t_i \in \mcl{T}$, the objective function~\eqref{problem5} of the dual problem~\eqref{dual} can be computed in polynomial time
as pointed out in Section~\ref{sec:singleuser}. Hence, we can apply a subgradient method to solve the problem~\eqref{dual} in polynomial time.

Let $\mbf{\tilde{R}}^{(t_i)}$ be the optimal rate tuple w.r.t. the problem~\eqref{problem5} for some weight vector $\mbf{\Lambda}^{(t_i)}$, $t_i \in \mcl{T}$.
Starting with a feasible iterate $\mbf{\Lambda}[0]$ w.r.t. the optimization problem~\eqref{dual},
every subsequent iterate $\mbf{\Lambda}[n]$ can be recursively represented as an Euclidian projection of the vector
\begin{align}
\mbf{\Lambda}_e[n] = \mbf{\Lambda}_e[n-1] + \theta [n-1]\cdot \mbf{\tilde{R}}_e[n-1],~~\forall e \in \mcl{E} \label{iterate}
\end{align}
onto the hyperplane $\left\{ \mbf{\Lambda}_e \geq \mbf{0} | \sum_{i=1}^k \lambda_e^{(t_i)} = \alpha_i \right\}$,
where $\mbf{\tilde{R}}_e[n-1] = \{R^{(t_i)}_e[n-1] : \forall t_i \in \mcl{T}\}$.
The Euclidian projection ensures that every iterate $\mbf{\Lambda}[n]$ is feasible w.r.t. the optimization problem~\eqref{dual}.
By appropriately choosing the step size $\theta[n]$ in each iteration, it is guaranteed that the subgradient method
converges to the optimal solution of the problem~\eqref{dual}.

To recover the primal optimal solution from the iterates $\mbf{\Lambda}[n]$ we apply the results from~\cite{sherali1996recovery},
where at each iteration $n$ of \eqref{iterate}, the primal iterate is constructed as follows
\begin{align}
\mbf{\hat{R}}[n] = \sum_{j=1}^n \mu_j^{(n)} \mbf{\tilde{R}}[j], \label{pr_recovery}
\end{align}
where
\begin{align}
\sum_{j=1}^n \mu_j^{(n)}=1,~\mu_j^{(n)}\geq 0,~\text{for}~j=1,2,\ldots,n.
\end{align}
By carefully choosing the step size $\theta[n]$, $\forall n$ in \eqref{iterate} and the convex combination coefficients $\mu_j^{(n)}$,
$\forall j=1,\ldots,n$, $\forall n$,
it is guaranteed that \eqref{pr_recovery} converges to the minimizer of \eqref{problem2}, and therefore to the minimizer
of the original problem~\eqref{problem1}. In~\cite{sherali1996recovery}, the authors proposed several choices for $\{\theta[n]\}$ and
$\{\mu_j^{(n)}\}$ which lead to the primal recovery. Here we list some of them.
\begin{enumerate}
\item $\theta[n]=\frac{a}{b+cn}$, $\forall n$, where $a>0$, $b\geq 0$, $c>0$,  \\
      $\mu_j^{(n)}=\frac{1}{n}$, $\forall j=1,\ldots,n$, $\forall n$,
\item $\theta[n]=n^{-a}$, $\forall n$, where $0<a<1$, \\
      $\mu_j^{(n)}=\frac{1}{n}$, $\forall j=1,\ldots,n$, $\forall n$.
\end{enumerate}

It is only left to compute an optimal \emph{MM}-rate vector w.r.t the linear objective defined in~\eqref{problem2}.
Let $\mbf{R}^{*}$ and $\mbf{Z}^{*}$ be the optimal rate vectors of the problems~\eqref{problem2} and~\eqref{problem4}, respectively.
As we pointed out $\mbf{R}^{*}=\mbf{Z}^{*}$, where $\mbf{Z}^{*}$ can be computed from $\mbf{\hat{R}}[n]$
for a sufficiently large $n$, as follows
\begin{align}
Z_e^{*} = \max \left\{ \hat{R}_e^{(t_1)}[n], \hat{R}_e^{(t_2)}[n], \ldots, \hat{R}_e^{(t_k)}[n] \right\},~~\forall e \in \mcl{E}. \nonumber
\end{align}

\subsection{Code Construction for the Linear Source Model}

In this Section we briefly address the question of the optimal code construction for the finite linear source model.
We begin our analysis by considering the following example.
\begin{example}
Consider a system with $k=2$ clients and $l=4$ source nodes presented in Figure~\ref{fig:twousers}.
For convenience, we express the data vector as
$\mbf{W}=\left[
           \begin{array}{cccc}
             a & b & c & d \\
           \end{array}
         \right] \in \mathbb{F}_{q^n}^4$, where $a,b,c,d$ are independent uniform random variables in $\mathbb{F}_{q^n}$.
Each source node has the following observations $\mbf{X}_{m_1}=\{a,b\}$,
$\mbf{X}_{m_2}=\{b,c\}$, $\mbf{X}_{m_3}=\{c\}$, $\mbf{X}_{m_4}=\{d\}$. Let the objective function be $\sum_{e\in \mcl{E}} R_e$,
with the capacity constraints $c_e=4$, $\forall e \in \mcl{E}$. Applying the algorithm described in this section, we obtain
\begin{align}
R^{*}_1=R^{*}_4=0,~R^{*}_2=R^{*}_5=1,~R^{*}_3=2,~R^{*}_6=3,~R^{*}_7=4. \nonumber
\end{align}

\end{example}
\begin{figure}[h]
\begin{center}
\includegraphics[scale=0.5]{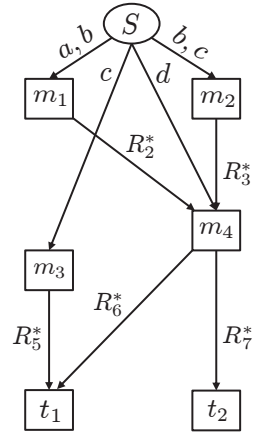}
\end{center}
\vspace{-0.2in}
\caption{Multicast network construction for the multisource multicast problem. We introduce a super node $S$ that posses
         all the data packets, and transmits them to the respective nodes.}\label{fig:codeconst}
\end{figure}

Now, we briefly explain how to design the actual transmissions of each source node. Starting from an optimal \emph{MM}-rate
vector, we first construct the corresponding multicast network by adding a super node $S$ that contains all individual
packets in $\mbf{W}$ (see Figure~\ref{fig:codeconst}). Then, we apply the algebraic network coding approach~\cite{KM03},
where the source matrix $\mbf{A}$ is given by
\begin{align}
\mbf{A}=\left[
          \begin{array}{cccc}
            \mbf{A}_{m_1}^T & \ldots & \mbf{A}_{m_l}^T & \mbf{0}_{N \times |\mcl{E}|} \\
          \end{array}
        \right].
\end{align}
Finally, the network code for the multisource multicast problem can be constructed in polynomial time from the algorithms
provided in~\cite{H05} which are based on a simultaneous transfer matrix completion.

In~\cite{KM03}, the authors derived
the transfer matrix $\mbf{M}(r_i)$ from the super-node $S$ to any receiver $t_i$, $i=1,\ldots,k$. It is a $|\mcl{E}| \times |\mcl{E}|$
matrix with the input vector $\mbf{W}$, and the output vector corresponding to the observations at the receiver $t_i$.
\begin{align}
\mbf{M}(t_i) = \mbf{A} (\mbf{I} - \mbf{\Gamma})^{-1} \mbf{B}(t_i),~~~i=1,\ldots,k,
\end{align}
where $\mbf{\Gamma}$ is adjacency matrix of the multicast network,
and $\mbf{B}(t_i)$ is an output matrix. For more details on how these matrices are constructed, we refer the interested reader to the reference \cite{KM03}.
Finally, given that $|\mathbb{F}_q|>k$,  the network code for the multisource multicast problem can be constructed in polynomial time from the algorithms
provided in~\cite{H05} which are based on a simultaneous transfer matrix completion\footnote{Complexity of the algorithm proposed in~\cite{H05} is
$\mcl{O}\left(k \cdot \left((|\mcl{E}|+N)^3\right) \log(|\mcl{E}|+N)\right)$.}.

\section{Conclusion}
In this work we study the linear cost multisource multicast problem, where each node in the network observes i.i.d. copies of the DMMS process. Assuming that nodes can communicate over the capacity constrained links of the directed acyclic graph,
the goal is for each client (sink of the graph), to learn the file, while minimizing a linear communication cost.
First, we show that the underlying optimization problem can be posed as a linear program with exponentially many rate-flow constraints.
Then, we provide the ``capacity flow'' conditions under which the multisource multicast problem is feasible. Applying the \emph{common base algorithm} one can construct a test that verifies feasibility in polynomial time. We show that the linear cost multisource multicast problem with single client and many nodes can be solved in polynomial time by applying algorithms for the minimum cost submodular flow problem.
Further, using the single client solution as a building block we show how one can solve a more general problem with arbitrary number of clients in polynomial time. For the special case of the finite linear source model, we propose a polynomial time algorithm that computes an explicit transmission scheme. 

\bibliographystyle{IEEEtran}

\end{document}